\documentclass[12pt, reqno]{amsart}
\usepackage{amsmath, amsthm, amscd, amsfonts, amssymb, graphicx, color}
\usepackage[bookmarksnumbered, colorlinks, plainpages]{hyperref}
\usepackage{enumerate}

\newtheorem{theorem}{Theorem}[section]
\newtheorem{lemma}[theorem]{Lemma}

\theoremstyle{definition}
\newtheorem{definition}[theorem]{Definition}
\newtheorem{example}[theorem]{Example}

\theoremstyle{remark}
\newtheorem{remark}[theorem]{Remark}
\numberwithin{equation}{section}

\begin{document}
\setcounter{page}{1}


\title[Volterra equations with piecewise continuous kernels]
{Solution to the Volterra integral equations of the first kind with piecewise continuous kernels
in class of Sobolev-Schwartz distributions}

\author[D.N.Sidorov]{Denis N. Sidorov$^{1,2}$ }

\address{$^{1}$ Department of Applied Mathematics, Melentiev Energy Systems Institute of 
Siberian Branch of Russian Academy of Sciences, Irkutsk 664033, Irkutsk, Russia.}
\address{$^{2}$ Institute of Mathematics, Economics and Informatics, Irkutsk State University, 664003, Irkutsk, Russia.}

\email{\textcolor[rgb]{0.00,0.00,0.84}{contact.dns@gmail.com}}



\date{Received: xxxxxx; Revised: yyyyyy; Accepted: zzzzzz.}

\begin{abstract}
Sufficient conditions for existence and uniqueness of the solution of the 
Volterra integral equations of the first kind with piecewise continuous kernels
are derived  in framework of Sobolev-Schwartz distribution theory.
The asymptotic approximation of the parametric
family of generalized solutions is constructed. The  method
for the solution's regular part refinement is proposed using 
the successive approximations method.
\end{abstract} \maketitle

\section{Introduction}

Let us define the triangular region $D = \{s,t; 0<s<t<T \}$
and introduce the  functions $s=\alpha_i(t), i=\overline{1,n},$ 
which are continuous and have continuous derivatives for
$t \in (0,T).$
We suppose $\alpha_i(0)=0,$ $0<\alpha_1(t)< \dots <
\alpha_{n-1}(t) < t$  for $t \in (0,T),$
$0< \alpha_1^{\prime}(0)< \dots < \alpha_{n-1}^{\prime}(0)<1,$
and functions $s=\alpha_i(t), \, i=\overline{0,n},$
 $\alpha_0(t)=0, \, \alpha_n(t)=t,$
 split the region $D$ into the following disjoint sectors 
 $D_1=\{s, t: 0\leq s <\alpha_1(t) \},$
$D_i = \{s,t : \alpha_{i-1}(t) < s < \alpha_i(t) , i=\overline{2,n}  \},$
$ \overline{D}=\overline{\bigcup\limits_1^n D_i}.$
Let us introduce the continuous functions  $K_i(t,s)$
defined for $t,s \in D_i,$ and differentiable wrt 
 $t,$
$i=\overline{1,n}.$

Let us consider the integral operator 
$$\int\limits_0^t K(t,s) u(s) ds \stackrel{\mathrm{def}}{=}  \sum\limits_{i=1}^n \int\limits_{\alpha_{i-1}(t)}^{\alpha_i(t)} K_i(t,s) x(s) ds 
\eqno{(1)}$$ 
with piecewise continuous kernels
$$    K(t,s) = \left\{ \begin{array}{ll}
         \mbox{$K_1(t,s), \,\,\, t,s \in D_1$}, \\
         \mbox{\,\dots \,\,\,\,\, \dots \dots} \\
         \mbox{$K_n(t,s), \,\,\, t,s \in D_n$}. \\
        \end{array} \right. \eqno{(2)} $$

In this paper we  deal with the following Volterra integral equation
$$\int\limits_0^t K(t,s) u(s)\,ds =f(t), \, 0<t<T\leq \infty, \eqno{(3)} $$
where function $f(t)$ 
 has a continuous derivative for $t \in (0,T), \, f(0) \neq 0.$
Equation (3) we call the Volterra integral equation (VIE) 
with piecewise continuous kernel.
Our objective is to construct the solution of VIE (3) in 
the space of Sobolev-Schwartz distributions [3]. 
Obviously, VIE (3) does not have
classic solutions since  $f(0) \neq 0.$

The differentiation of VIE (3) leads to integral-functional 
equation and its solution in not unique in the general case [16]. 
That is why  study of VIE (3) can not be performed 
using only the classic methods in the Volterra theory 
[1,6,15,14].
In this paper we continue our results on VIE studies
[11], [8,9], [17,19]. We consider the equation (3) 
 using the elementary results of the theory of  integral and difference equations,
functional analysis, Sobolev-Schwartz distributions and
theory of functional equations with perturbed argument of 
neutral type [7]. 

This paper is organized as follows.
Section 2 outlines the construction of the singular component of the solution
and the integral-functional equation for the regular component of the solution
is derived. 
In section 3 we obtain the sufficient conditions for existence and
uniqueness of solution of VIE (3) in the following form $u(t) = a \delta(t)+x(t),$
were $\delta(t)$ is Dirac delta function, $x(t)$ is regular continuous 
function.
Such solutions satisfy to the equation (3) in sense of Sobolev-Schwartz distributions [3].
To the best of our knowledge, similar studies on VIE (3) have not yet been reported in literature.
In section 3 we construct the regular part of the solution using the 
the``step method'' [13] from the theory of functional equations and successive 
approximations method.
In sections 4 and 5 we address the most interesting case when 
VIE (3) has family of  solutions depending on free parameters.
  The method for construction of asymptotic
approximations of parametric solutions is proposed 
and iterative refinement method is constructed.
It is to be noted that known method of  A.O.~Gelfond (readers may refer to [4],\, p.338) 
of solution of difference equations is employed.

\section{Definition of the singular component of the solution}

Let us extend $f(t)$ on negative semiaxis with zero and differentiation of VIE (3) 
yields the following equivalent equation
$$F(u)\stackrel{\mathrm{def}}{=} K_n(t,t)u(t) + \sum\limits_{i=1}^{n-1} \alpha_i^{\prime}(t)
\biggl\{ K_i(t,\alpha_i(t)) -  \eqno{(4)}$$
$$-K_{i+1}(t,\alpha_i(t)) \biggr\} u(\alpha_i(t)) + \sum\limits_{i=1}^n \int\limits_{\alpha_{i-1}(t)}^{\alpha_i(t)} K_i^{(1)}(t,s) u(s) \,ds =  f^{(1)}(t) + f(0)\delta(t),  $$
where $\alpha_0=0, \, \alpha_n(t)=t.$
Le us assume   $K_1(0,0)\neq 0, \, K_n(t,t)\neq 0,  $
for $t\in [0,T].$
Let us introduce the following functional operator
$$Au \stackrel{\mathrm{def}}{=} \sum\limits_{i=1}^{n-1} K_n^{-1}(t,t)
\alpha_i^{\prime}(t) \{ K_i(t,\alpha_i(t)) - K_{i+1}(t,\alpha_i(t)) \} u(\alpha_i(t))  $$
and integral operator
$Ku \stackrel{\mathrm{def}}{=}  \sum\limits_{i=1}^n \int\limits_{\alpha_{i-1}(t)}^{\alpha_i(t)} K_n^{-1}(t,t)
K_i^{(1)}(t,s) u(s) ds.$

Then equation (4) can be reduced to the following  equation
$$u(t) + Au + Ku = K_n^{-1}(t,t) f^{(1)}(t) + K_n^{-1}(0,0)f(0) \delta(t). \eqno{(5)}$$
Let us search for a solution to VIE (3.2) of the form 
$u(t) = a \delta(t) + x(t),$
where $a $ is constant,  $x(t) \in \mathrm{C}_{(0,T)}.$

It is easy to verify the following identities:
$$ \int\limits_0^{\alpha_1(t)} \frac{\partial K_1(t,s)}{\partial t} \delta(s) ds =
\frac{\partial K_1(t,0)}{\partial t},$$
$$\int\limits_{\alpha_{i-1}(t)}^{\alpha_i(t)} \frac{\partial K_i(t,s)}{\partial t} \delta(s) ds = 0 $$
for $i=\overline{2,n}.$
Indeed, the first identity holds because
 $\alpha_1(t)>0,$ $\frac{\partial K_1(t,s)}{\partial t} \delta(s) = \frac{\partial K_1(t,0)}{\partial t}
\delta(s),$ $\int\limits_0^{\alpha_1(t)} \delta(s) ds = \theta(\alpha_i(t)) = 1$
for $t>0,$
were $\theta $ is Heaviside function.
The second identity is becomes trivial if we notice that for 
$i=\overline{2,n}$
 $supp \, \delta(s) \cap D_i =0,$
$\int\limits_{\alpha_{i-1}(t)}^{\alpha_i(t)} \delta(s) ds = \theta (\alpha_i(t)) - \theta(\alpha_{i-1}(t)) = 0,$
since $0<\alpha_1(t) < \alpha_2(t) < \dots < \alpha_n(t) = t.$
Let us also recall the identity $\delta(\alpha_i(t)) = \frac{\delta(t)}{|\alpha_i^{\prime}(0)|}$  (refer, e.g. to [3],  p. 34).
Let us take into account the outlined identities and substitution  $u = a \delta(t) + x(t)$
leads the equation (5) to the following equation
$K_n^{-1}(0,0) K_1(0,0) a \delta(t) + K_n^{-1}(t,t) \frac{\partial K_1(t,0)}{\partial t} a +
x(t) + Ax + Kx = K_n^{-1}(t,t) f^{(1)}(t) + K_n^{-1}(0,0)f(0) \delta(t).$
Equating the last equation coefficients of $\delta(t)$ result
$a = \frac{f(0)}{K_1(0,0)}.$  
It is remains to determine the regular part
from the equation
$$
x(t) + A x + K x = \overline{f}(t),
\eqno{(6)}
$$
where $\overline{f}(t) = K_n^{-1}(t,t) \biggl\{ f^{(1)}(t) - \frac{\partial K_1(t,0)}{\partial t} 
\frac{f(0)}{K_1(0,0)}\biggr\}.$
It is to be noted that due to the operator equality 
 $$K_n(t,t)(I + A + K)x = F(x) $$
the equation (6) can be written as follows
$$F(x) = f^{\prime}(t) - \frac{\partial K_1(t,0)}{\partial t} \frac{f(0)}{K_1(0,0)}. 
\eqno{(7)}$$

\section{Sufficient conditions for the existence of a unique generalized solution}

Since $K_1(0,0) \neq 0$ then homogeneous equation (4)
has only trivial solution of singular functions
 $$u_{sing} \stackrel{\mathrm{def}}{=}  \sum\limits_0^m c_i \delta^{(i)}(t)$$
with support at the origin.  Therefore, the existence and uniqueness of generalized 
solutions of the equation (4) 
$$u(t)=u_{sing} + x(t), $$  $x(t) \in \mathbb{C}_{(0,T)}$
is equivalent to proving the existence of
unique solution of equation (6) in 
 $\mathbb{C}_{(0,T)}.$
Let us introduce the function
$$\,\,\,\,\,\,\,\,\,\,\,\,\,\,|A(t)|  \stackrel{\mathrm{def}}{=}  \biggl|\sum\limits_{i=1}^{n-1}
\alpha_i^{(1)}(t) K_n^{-1}(t,t) \biggl\{ K_i(t,\alpha_i(t)) - K_{i+1}(t,\alpha_i(t)) \biggr\}\biggr|. \,\,\,\,\,\,\,\,\,\,\,\,\,\,\,\,\,\,\,\,\, (\ast) $$
Let the following condition be fulfilled\\
\begin{enumerate} [{\bf (A)}]
\item $| A(0) |<1,$  $\sup\limits_{0<s<t<T} |K_n^{-1}(t,t)K(t,s) | \leq c <\infty.$
\end{enumerate}
Condition (A) is fulfilled if $\alpha_i^{(1)}(0)$ are sufficiently small. 
Here and below the kernel $K(t,s)$  in $\bigcup\limits_1^n D_i$ 
is defined as (2). It's derivative wrt $t$ for $t,s \in \bigcup\limits_1^n D_i$ 
is defined as follows:
$$    K^{(1)}(t,s) = \left\{ \begin{array}{ll}
         \mbox{$K_1^{(1)}(t,s), \,\,\, t,s \in D_1$}, \\
         \mbox{\,\dots \,\,\,\,\, \dots \dots} \\
         \mbox{$K_n^{(1)}(t,s), \,\,\, t,s \in D_n$}. \\
        \end{array} \right.   $$

\begin{theorem} { (Sufficient conditions for existence and uniqueness of generalized solutions).}

Let the condition { (A)} be fulfilled, $K_i(t,s)$ in (2) are continuous
functions, are they have continuous derivatives wrt $t,$
function $f(t)$ has continuous derivative, $f(0)\neq 0.$
Let $K_1(0,0) \neq 0.$
Then equation (3) has the unique solution
$$u(t) = \frac{f(0)}{K_1(0,0)} \delta(t) +x(t),$$
where $x(t) \in \mathbb{C}_{(0,T)}.$ At the same time we can find $x(t)$ 
using the step method combined with
successive approximations.
\end{theorem}

\begin{proof} 

Since the singular part of the solution is defined 
then let us consider the equation (6)  satisfied by the regular
component $x(t).$

Let us fix $q<1$ and select $h_1>0$ such as $\sup\limits_{0\leq t \leq h_1}|A(t)| = q<1.$
Due to the condition {(A)} such a variable $h_1>0$ exists.
Let $0<h<\min \{ h_1, \frac{1-q}{c} \},$
where variable $c$ is defined in condition {(A)}.
Let us divide the interval $[0,T]$ into subintervals
$$
[0,h], \, [h, h+\varepsilon h], \, [h+\varepsilon h, h+ 2\varepsilon h], \dots .
\eqno{(8)}
$$
We denote by $x_0(t)$ the restriction of the solution $x(t)$
into $[0,h],$ and by $x_m(t)$ we denote it's restriction into  subintervals 
$$I_m = [(1+(m-1)\varepsilon)h, (1+m\varepsilon)h],\, m=1,2,\dots .$$
Let us select $\varepsilon$  from $(0,1]$ such as for $t\in I_m$ 
``perturbed'' arguments $\alpha_i(t) \in \bigcup\limits_{k=1}^{m-1} I_k, \, i=\overline{1,n-1}.$
If $0<\alpha_i^{(1)}(t) < \frac{1}{1+\varepsilon}$
for $t \in [0,T), \, i=\overline{1,n-1},$ 
then the above inclusion holds in the interval $[0,T).$ 
This inclusion makes it possible to apply the well-known 
in the theory of functional differential equations ``step method''.
The readers may refer to [13],\, p. 199.

Let us construct the sequence $\{x_0^n(t)\}:$
$$x_0^n (t) = -A x_0^{n-1} - K x_0^{n-1} + \overline{f} (t), $$
$$x_0^0(t) = \overline{f}(t), \, t \in [0,h]. $$
to define $x_0(t) \in \mathbb{C}_{[0,h]}$

Due to the selection of $h$ we have an estimate $||A+K||_{\mathcal{L}(\mathbb{C}_{(0,h)} \rightarrow \mathbb{C}_{(0,h)})} < 1.$

Therefore for $t\in [0,h]$ exist  unique solution $x_0(t)$ of equation (6). 
The sequence 
$x_0^n(t)$ uniformly converge to the solution.
We continue the process of constructing the desired solution for
 $t\geq h,$ i.e. on the intervals 
$I_n, \, n=1,2, \dots .$
For sake of clarity let $\varepsilon = 1$ in (8).

Once we get the element  $x_0(t) \in \mathbb{C}_{[0,h]}$
computed we will look for an element $x_1(t) $ 
in the space $\mathbb{C}_{(h,2h)}$.
We will find $x_1(t)$ from the Volterra integral
equation of the 2nd  kind
$$x(t) + \int\limits_h^t K_n^{-1}(t,t) K_t^{\prime}(t,s) x(s) \,ds = \overline{f}(t) - A x_0 - \int\limits_{0}^h K_n^{-1}(t,t)
K_t^{\prime}(t,s) {x}_0(s) \,ds$$
using the successive approximations, with $x_0(h) = x_1(h).$

Let us introduce the continuous function
 $$
    \overline{x}_1(t) = \left\{ \begin{array}{ll}
         \mbox{$x_0(t), \,\,\,\, {0} \leq t \leq h$}, \\
         \mbox{$x_1(t), \,\, h \leq t \leq 2h $}, \\
\end{array}          
\right.
\eqno{(9)}
 $$
which is the reduction of continuous solution $x(t)$  on to $[0,2h].$
Then we can find element $x_2(t) \in$ $\mathbb{C}_{(2h,3h)}$
using the successive approximations 
from the Volterra integral equation of the 2nd kind
$$x(t) + \int\limits_{2h}^t K_n^{-1}(t,t) K_t^{\prime}(t,s) x(s) \, ds = $$ $$= \overline{f}(t) - A \overline{x}_1
- \int\limits_{0}^{2h} K_n^{-1}(t,t) K_t^{\prime}(t,s) \overline{x}_1(s)\, ds.$$
Continuing this process for $N$ steps ($ N \geq \frac {T} {h} $) 

We will construct the desired solution $x(t) \in \mathbb{C}_{(0,T)}$ of VIE (3)
using this process for $N$ steps.

This completes the proof of the theorem.
\end{proof} 

\section{Construction of an asymptotic approximation $\hat{x}(t)$ 
of the regular part of parametric family of the desired solution }

Let us consider the equation (7) which is satisfied by the regular part of
generalized solution.
Let the following condition be fulfilled
\begin{enumerate} [{\bf (B)}]
\item Exist  polynomials $\mathcal{P}_i = \sum\limits_{\nu+\mu = 1}^N K_{i \nu \mu} t^{\nu} s^{\mu}, \, i=\overline{1,n},$\\  $f^N(t) = \sum\limits_{\nu=1}^N f_{\nu} t^{\nu},$
 $\alpha_i^N(t) = \sum\limits_{\nu = 1}^N \alpha_{i \nu} t^{\nu}, i=\overline{1,n-1},$
where
$0< \alpha_{11} < \alpha_{21}  < \dots < \alpha_{n-1,1} <1,$
such as for $t\rightarrow +0,$ $s\rightarrow +0$ we have the following estimates
$| K_i(t,s) -\mathcal{P}_i(t,s) | = \mathcal{O}((t+s)^{N+1}), \, i=\overline{1,n},$ 
$|f(t) - f^N(t)| = \mathcal{O}(t^{N+1}),$  $|\alpha_i(t) - \alpha_i^N(t)| = \mathcal{O}(t^{N+1}), i=\overline{1,n-1}.$

\end{enumerate}

Expansion in powers of $t, s$ which are presented in the condition {(B)} we will call as
  ``Taylor polynomials'' of the corresponding functions.
Let us introduce the function
$$
B(j) = K_n(0,0) + \sum\limits_{i=1}^{n-1} (\alpha_i^{\prime}(0))^{1+j}(K_i(0,0)-K_{i+1}(0,0)),
$$
which depends on  integer argument $j,$ $j \in \mathbb{N} \cup {0}.$ Function
 $B(j)$ which corresponds to the main ``functional'' part of the equation (7) 
is called as {\it characteristic function} of equation (7).
Let us consider the construction of an asymptotic solution of equation (7).

In contrast to section 3, in sections 4 and 5 it is not supposed that 
homogeneous 
 equation for equation (3) has only trivial solution. 
Therefore the solution of integral-functional
equation (7) can be non unique. Let us follow paper  [5] and search 
for the asymptotic approximation of 
 a particular solution of the inhomogeneous equation
 (7)
as following polynomial
$$
\hat{x}(t) = \sum\limits_{j=0}^N x_j (\ln t) t^j.
\eqno{(10)}
$$
Let us demonstrate that coefficients  $x_j$ depends on $\ln t$ 
and free parameters in  general irregular  case. This is consistent with the possibility of the existence of nontrivial solutions of the homogeneous equation.

For computation of the coefficients $x_j$ we consider regular and irregular cases.
\begin{definition}
Point $j^*$ is called  {\it regular  point} of characteristic function $B(j),$
if $B(j^*) \neq 0$ and {\it irregular point}  otherwise.
\end{definition}

\subsection{The regular case: characteristic function $B(j) \neq 0$  for  $j\in (0,1, \dots ,N),$  where $N$ is sufficiently  large }

In this case, the coefficients of  $x_j$ will be constant, i.e. independent on $ \ln t$.
Indeed, lets substitute the expansion (10) into the equation (7). Using the method of undetermined coefficients and taking into account the conditions {(B)},
lead  to the recursive sequence of the systems of linear algebraic
equations wrt  $x_j:$
$$
B(0)x_0 = f^{\prime}(0) - \frac{f(0)}{K_1(0,0)} - \frac{f(0)}{K_1(0,0)} \frac{\partial K_1(t,0)}{\partial t} \biggr|_{t=0},
\eqno{(11)}
$$
$$
B(j)x_j = M_j(x_0, \dots , x_{j-1}), \, j=1,\dots , N.
\eqno{(12)}
$$
 $M_j$  are expressed in terms of solutions
$x_0, \dots , x_{j-1}$ of previous equations and coefficients of
the Taylor polynomials from the condition {(B)}.

Since in the regular case  $B(j) \neq 0$ the coefficients
$x_0, \dots , x_N$ can be uniquely determined and  the asymptotics (10) can be constructed
by this means.

\subsection{Irregular case: characteristic function $B(j)$ in  $(0,1,\dots,N)$ has zeros}

Let us demonstrate that in irregular case the coefficients $x_j$
are polynomials in powers of $\ln t$ and depends upon
arbitrary constants. The order of polynomials and the number of arbitrary constants are related to the multiplicities of integer solutions of the equation $B(j) = 0$.

Indeed, since the coefficient $x_0$ in the irregular case
can depend on  
$\ln t,$ then based on the method of undetermined 
coefficients $x_0$  can be found as the solution of the difference equation
  $$
K_n(0,0) x_0(z) + \sum\limits_{i=1}^{n-1} \alpha_i^{\prime}(0) (K_i(0,0) - K_{i+1}(0,0)) x_0(z+a_i) = \eqno{(13)}$$
$$ = f^{\prime} (0) - \frac{f(0)}{K_1(0,0)} \frac{K_1(t,0)}{\partial t}\biggr|_{t=0},
$$
where $a_i = \ln \alpha^{\prime}(0), z = \ln t.$
There are three possible cases here:

\noindent {\it 1st case.} $(B(0)\neq 0).$ \\  
In this case the coefficient $x_0$ does not depend on $z$
and can be determined uniquely from the equation (11).

\noindent {\it 2nd case.} $(B(0)=0).$\\
Let $j=0$ is simple zero of the function $B(j),$ 
i.e. $B(0)=0, \, B^{\prime}(0) \neq 0.$  Then the coefficient $x_0(z)$
we can find from the difference equation (13)
as linear function
$$
x_0(z) = x_{01} z + x_{02}. \eqno{(14)}
$$
Lets substitute (14) into  (13). Thus for determination of the coefficients $ x_{01}, x_{02} $
we obtain two equations as follows:
$$
B(0) x_{01} = 0,
\eqno{(15)}
$$
$$
B(0) x_{02} + B^{(1)}(0) x_{01}  =  f^{\prime}(0) - \frac{f(0)}{K_1(0,0)} \frac{\partial K_1(t,0)}{\partial t} \biggr|_{t=0}, \eqno{(16)}
$$
where $B(0)=0, \, B^{(1)}(0)\neq 0.$
Hence the coefficient $x_0(z)$ is linear wrt  $z$ and depends on the arbitrary constant.
So, it the 2nd case $$x_0(z) = \biggl ( f^{(1)}(0) - \frac{f(0)}{K_1(0,0)} \frac{\partial K_1(0,0)}{\partial t} \biggr) \frac{1}{B^{(1)}(0)} z +c, \,$$  where $c$ is arbitrary constant. 

\noindent {\it 3rd case. } Let $j=0$ be root of the equation $B(j)=0$ with order of multiplicity of $k+1,$
i.e. $B(0)=B^{\prime}(0) = \dots B^{(k)}(0)=0, \, B^{(k+1)}(0) \neq 0,$
 $k \geq 1.$ Solution $x_0(z)$ of the difference equation (12)
we search in the form of a polynomial
$$
x_0(z) = x_{01} z^{k+1} + x_{02} z^k + \dots + x_{0 k+1} z + x_{0 k+2}.
\eqno{(17)}
$$
Let us substitute the polynomial (17) into the equation (13)
and take into account the identity $$\frac{d^k}{dj^k}B(j) = \sum\limits_{i=1}^{n-1} (\alpha_i^{\prime}(0))^{1+j} a_i^k (K_i(0,0) - K_{i+1}(0,0)),$$
where $a_i = \ln \alpha_i^{\prime}(0).$
Next lets equate the coefficients of powers $$ z^{k+1}, z^k, \dots , z, z^0$$ to zero.
Finally we  get recurrent sequence of linear algebraic equations wrt $x_{01}, x_{02}, \dots , x_{0 k+2}:$
$$ \left\{ \begin{array}{ll} 
B(0) x_{01} = 0,\\
B(0) x_{02} + B^{(1)}(0) 
\biggl(\begin{array}{c}
	k+1\\ k
	\end{array}
	\biggr) x_{01} = 0,  \,\,\,\,\,\,\,\,\,\,\,\,\,\,\,\,\,\,\,\,\,\,\,\,\,\,\,\,\,\,\,\,\,\,\,\,\,\,\,\,\,\,\,\,\,\,\,\,\,\,\,\,\,\,\,\,\,\,\,\,\,\,\,\,\,\,\,\,\,\,\,\,\,\,\,\,\,\,\,\,\,\,{(18)}\\
	B(0) x_{0 l+1} + B^{(l)}(0) \biggl(\begin{array}{c}
	k+1\\ k+1-l
	\end{array}
	\biggr) x_{01} + B^{(l-1)}(0) \biggl(\begin{array}{c}
	k\\ k+1-l
	\end{array}
	\biggr) x_{02} + \dots  \\ \dots + B^{(1)}(0) 
	\biggl(\begin{array}{c}
	k+1-l+1\\ k+1-l
	\end{array}
	\biggr) x_{ol} = 0, \, l=1,\dots , k,
\end{array} \right . 
$$
$$ B(0) x_{0 k+2} + B^{(k+1)}(0) x_{01} + B^{(k)}(0) x_{02} + \dots B^{(1)}(0) x_{0 k+1} =
\eqno{(19)}
$$
$$ = f^{\prime}(0) - \frac{f(0)}{K_1(0,0)} \frac{\partial K_1(t,0)}{\partial t} \biggr|_{t=0}.
$$
In our case  $B(0) = B^{\prime}(0) = \dots = B^{(k)}(0)=0, \, B^{(k+1)}(0) \neq 0.$ 
Hence in polynomial (17) we let
  $x_{01} = \frac{1}{B^{(k+1)}(0)} \biggl ( f^{\prime}(0) -\frac{f(0)}{K_1(0,0)} \frac{\partial K_1(0,0)}{\partial t} \biggr ).$
Equations of the system (18)  become identities
$B(0)x_{0j}=0, \, j=\overline{1,k+1},$  $B(0)=0.$
Hence the coefficients $x_{02}, \dots , x_{0 k+2}$
 of the polynomial (17) remains arbitrary constants.

Next, let's employ the method of undetermined coefficients and take into account the identity
$$\int t^j \ln^k t\, dt = t^{j+1} \sum\limits_{s=0}^k (-1)^s \frac{k(k-1) \dots (k-(s-1))}{(j+1)^{s+1}} \ln^{k-s} t. $$
 By this means we construct the difference equations for determination of the coefficient
$x_1(z)$  ($z=\ln t$) and next coefficients of the asymptotics (10).
Indeed,
$$
 L(x)\biggr |_{x=x_0(z)+x_1(z) t} \stackrel{\mathrm{def}}{=} \biggl[ K_n(0,0)x_1(z) + \sum\limits_{i=1}^{n-1}
(\alpha_i^{\prime}(0))^2 (K_i(0,0) -
\eqno{(20)}
$$ 
$$-K_{i+1}(0,0)) x_1(z+a_i) + P_1(x_0(z)) \biggr] t + r(t), \,\,\,\,\, r(t) = o(t).$$
Here $P_1(x_0(z))$ is the polynomial of  $z,$ 
which degree is equal to the
multiplicity of solution $j=0$ of  the equation $B(j)=0$
as have been proved.
From the relation (20)  due to $r(t)=o(t)$
for $t\rightarrow 0$ it follows that coefficient $x_1(z)$
have to satisfy the difference equation
$$
K_n(0,0) x_1(z) + \sum\limits_{i=1}^{n-1} (\alpha^{\prime}(0))^2 \bigl(K_i(0,0) - K_{i+1}(0,0)\bigr) x_1(z+a_i)+
\eqno{(21)}
$$
$$+ P_1(x_0(z)) = 0. $$
If $B(1)\neq 0,$ then the equation (21)
has the solution $x_1(z)$ as the same degree polynomial as multiplicity order of the solution $j=0$
of the equation $B(j)=0.$
If $j=1$ is also the solution of the equation $B(j)=0$ the solution $x_1(z)$
can be constructed as polynomial of the power $k_0 +k_1,$
where $k_0$ and $k_1$ are the corresponding multiplicities of the solutions $j=0$ and $j=1$
of equation $B(j)=0$. Coefficient  $x_1(z)$ depends on
$k_0 +  k_1$
arbitrary constants.

Let us introduce the condition
\begin{enumerate} [{\bf (C)}]
\item
Let the equation $B(j)=0$ in the array $(0,1, \dots , N)$
has the solutions  $j_1, \dots , j_{\nu}$ of multiplicities
 $k_i,\, i=\overline{1,\nu}.$
\end{enumerate}

Then, in a similar way we can calculate the remaining coefficients
$x_2(z), \dots , x_N(z)$ of the asymptotic approximation $\hat{x}(t)$
of solution of equation (7) from the following  sequence of difference 
equation
$$K_n(0,0)x_j(z) + \sum\limits_{i=1}^{n-1} (\alpha^{\prime}(0))^{1+j}\bigl(K_i(0,0)-K_{i+1}(0,0)\bigr)x_j(z+a_i)+$$
$$+\mathcal{P}_j(x_0(z), \dots , x_{j-1}(z)))=0, \, j=\overline{2,N}.$$
Thus we have the following lemma\\
\begin{lemma}
Let  conditions {(B)} and {(C)} be fulfilled. Then exists the function
$\hat{x}(t) = \sum\limits_{i=0}^N x_i (\ln t) t^i,$ such as for
$t \rightarrow +0$  the residual solution of equation (7) satisfies the estimate
$$\biggl|F(\hat{x}(t)) - f^{(1)}(t) + K^{(1)}(t,0) \frac{f(0)}{K_1(0,0)}\biggr| = o(t^N).$$

The coefficients $ x_i (\ln t) $ are polynomials of $ \ln t.$
The degrees of these polynomials are increasing  and do not exceed the sum of the multiplicities of $ \sum \limits_j k_ {j} $ of solutions of the equation $ B(j) = 0 $ from the array $ (0,1, \ dots, i). $
The coefficients $ x_i (\ln t) $ depend on $ \sum_ {j = 0} ^ i k_j $
arbitrary constants.

\end{lemma}

\begin{remark}
If $B(j)\neq 0,$ then in the sum $\sum\limits_{j=0}^i k_j$ we zero the corresponding $k_j.$ 
\end{remark}

\section{An existence theorem for continuous parametric solutions families}

Since $0 < \alpha_i^{\prime}(0) <1, \, \alpha_i(0)=0, \, i=\overline{1,n-1},$
then for any  $0<\varepsilon <1$
exists $T^{\prime} \in(0,T]$
such as the following estimates are fulfilled $$\max\limits_{i=\overline{1,n-1}, t\in [0,T^{\prime}]} |\alpha_i^{\prime} (t)| \leq \varepsilon,$$
 $$\sup\limits_{i=\overline{1,n-1}, t\in (0,T^{\prime}]} \frac{\alpha_i(t)}{t} \leq \varepsilon.$$

Let us introduce the condition

\begin{enumerate} [{\bf (D)}]
\item Let $K_n(t,t) \neq 0$ for  $t\in [0,T^{\prime}]$ and $N^*$
is chosen so large that the following equality
$$\sup\limits_{t\in (0,T^{\prime})} \varepsilon^{N^*} |A(t)|  \leq q <1
$$
is fulfilled, where function $A(t)$ is defined the Section 3 with formula ($\ast$).
\end{enumerate}
 
\begin{lemma}
Let the condition { (D)} be fulfilled. Let in $\mathbb{C}_{(0,T^{\prime})}$
class of continuous functions for $t \in (0,T^{\prime}]$
which have the limit (which could be infinite) for $t \rightarrow +0$
exists an element $\hat{x}(t)$ such as for $t\rightarrow +0$
error of the solution of equation (7) satisfy the estimate
$$
\bigl|F(\hat{x}(t)) - f^{\prime}(t) + K_1^{\prime}(t,0) \frac{f(0)}{K_1(0,0)}\bigr| = o(t^N), $$
where $N \geq N^*.$
Then equation (7) in $\mathbb{C}_{(0,T^{\prime})}$
has the solution
$$x(t) = \hat{x}(t) +t^N v(t),  \eqno{(22)}$$
where $v(t)$ is uniquely determined by successive approximations.
\end{lemma}

\begin{proof}
Substitution of (22) in equation (7) gives us 
the following integral-functional equation for 
determination of the function 
$v(t)$ 
$$
v(t) + K_n(t,t)\biggl\{ \sum\limits_{i=1}^{n-1} \alpha_i^{\prime}(t)
\biggl(\frac{\alpha_i(t)}{t} \biggr)^{N^*}  \biggl( K_i(t, \alpha_i(t))  -
\eqno{(23)}
$$
$$
-K_{i+1}(t,\alpha_i(t)) \biggr) v(\alpha_i(t)) +\sum\limits_{i=1}^n
\int\limits_{\alpha_{i-1}(t)}^{\alpha_{i}(t)} K_i^{(1)}(t,s) \biggl(\frac{s}{t}\biggr)^{N^*} v(s) \, ds
\biggr\}  = $$ $$ =\biggl\{f^{\prime} (t) - \frac{\partial K_1(t,0)}{\partial t} \frac{f(0)}{K_1(0,0)} -
F(\hat{x}(t)) \biggr\} (t^{N^*}K_n(t,t))^{-1}.
$$
Let us introduce the linear operators
$$
Mu \stackrel{\mathrm{def}}{=}  K_n^{-1}(t,t)  \sum\limits_{i=1}^{n-1}
\alpha_i^{\prime}(t) \biggl(\frac{\alpha_i(t)}{t} \biggr)^{N^*}
\biggl \{ K_i(t, \alpha_i(t))  -$$
$$-K_{i+1}(t,\alpha_i(t)) \biggr\} v(\alpha_i(t)), $$
$$K v \stackrel{\mathrm{def}}{=} \sum\limits_{i=1}^n \int\limits_{\alpha_{i-1}(t)}^{\alpha_i(t)}
K_n^{-1}(t,t) K_i^{(1)}(t,s) (s/t)^{N^*} v(s)\, ds. $$
Then equation  (23) can be presented as following operator equation
$$u+ ( M+K) u  = \gamma(t), $$
where $\gamma(t) $ is the right hand side of the equation (23). This function
is continuous due to the condition of the Lemma 2.
Let us introduce the Banach space $X$ of continuous
functions  $v(t)$ with norm $$ ||v||_l = \max\limits_{0\leq t \leq T^{\prime}} e^{-lt} | v(t)|, \, l>0.$$
Then due to the inequalities $\sup\limits_{t \in (0,T^{\prime}]} \frac{\alpha_i(t)}{t} \leq \varepsilon <1$
and due to the condition { (D)} for $\forall l \geq 0$ 
norm of a linear function of the operator $ M $ satisfies
$$||M ||_{\mathcal{L}(X \rightarrow X)} \leq q < 1.$$
In addition, for the integral operator $K$
for sufficiently large  $l$ the following estimate is correct
$$||K ||_{\mathcal{L}(X \rightarrow X)} \leq q_1 < 1-q.$$
For sufficiently large $l>0$ this implies that
$$||M+K ||_{\mathcal{L}(X \rightarrow X)} < 1,$$
i.e. the linear operator $M+K$ 
is a contraction operator in the space $X.$
Hence the sequence $\{v_n \}$ converge
where $v_n = -(M+K)v_{n-1} + \gamma(t), \, v_0=\gamma(t)$.
This completes the proof of the theorem.
\end{proof}

   \begin{theorem}
Let the following conditions be fulfilled { (B)}, { (C)}, { (D)}, $f(0)\neq 0$, $K_1(0,0) \neq 0.$ 
Then equation (3) for $0<t\leq T^{\prime} \leq T$ has the solution
$$ x(t) = \frac{f(0)}{K_1(0,0)}\delta(t) + \hat{x}(t) + t^{N^*} v(t), $$
which depends on $\sum\limits_{i=1}^{\nu}  k_i$
arbitrary constants, where  $k_i$ are determined in the condition (C).
Function $\hat{x}$ is 
constructed in the form of  (10), then $v(t) $ is uniquely determined
with successive approximations. And we have the following asymptotic estimate 
 $\biggl|x(t) - \frac{f(0)}{K_1(0,0)} \delta(t) - \hat{x}(t)\biggr| = \mathcal{O}(t^{N^*})$
for $t\rightarrow +0.$
\end{theorem}

\begin{proof}
Based on the Lemma 1 because of the conditions of the theorem is possible to construct an asymptotic approximation of the regular part
$ \hat{x}(t) $ of the solution in the form of the following log-power polynomial:
 $$\sum\limits_{i=0}^N x_i (\ln t) t^i.$$
In this case, by construction, the coefficients $x_i(\ln t)$ 
depend on the certain number of arbitrary constants.
Due to the lemma 2 the substitution $x(t) = \hat{x}(t)+
t^{N^*} u(t)$ enable the construction of the continuous function 
$u(t)$ using the successive approximations method. 
\end{proof}

The solution constructed on
$[0,T^{\prime}]$  can be extended on the whole interval 
$[0,T],$ based on known ``step method'' [13, c.199].

In simple cases one can use the solution of  the equivalent equation (4)
in order to construct the 
solution of the integral equation (3) in closed form.

\begin{example}
$$\int\limits_0^{t/2} x(s) ds + 2 \int\limits_{t/2}^{t} x(s) ds = 2+t, \, t>0.$$
An equivalent equation (4) in this example
has the following form
$ -\frac{t}{2} x(\frac{t}{2}) + 2 x(t) = 2\delta(t) + 1.$
The desired solution is as follows
  $x(t) = 2 \delta(t) + 2/3.$
\end{example}
  
\begin{example}
$$\int\limits_0^{t/2} x(s) ds - \int\limits_{t/2}^{t} x(s) ds = 1+t, \, t>0.$$
Here an equivalent equation is 
as follows
$ x(\frac{t}{2}) - x(t) = \delta(t) + 1.$
It has  $c$--parametric family of generalized solutions
  $x(t) = \delta(t) + c - \frac{\ln t}{\ln 2}, \, c $ is constant.
\end{example}

\end{document}